\PassOptionsToPackage{american}{babel}
\documentclass[11pt]{article}
\usepackage{microtype}
\usepackage{mathtools}
\usepackage{dashrule}
\usepackage[inline]{enumitem}
\usepackage{amssymb}
\usepackage{amsthm}
\usepackage{amsmath}
\usepackage{fullpage}
\usepackage{subcaption}
\usepackage{hyperref}
\newtheorem{theorem}{Theorem}[section]
\newtheorem{lemma}[theorem]{Lemma}
\newtheorem{definition}[theorem]{Definition}

\newcommand{\eps}{\varepsilon}
\newcommand{\ball}{\mathrm{B}}
\newcommand{\lca}{\mathrm{lca}}

\newcommand*{\sub}[1]{{#1}_{S}}
\newcommand*{\pr}[1]{{#1}^\prime}
\newcommand*{\gclus}[2]{\textsc{Generalized} \ensuremath{({#1},{#2})}-\textsc{Clustering}}
\newcommand*{\lclus}[2]{\textsc{Local} \ensuremath{({#1},{#2})}-\textsc{Clustering}}
\newcommand*{\clus}[2]{\textsc{Labeled} \ensuremath{({#1},{#2})}-\textsc{Clustering}}
\newcommand\NP{\textsc{NP}}
\def\NIL/{{\sc Nil}}
\DeclarePairedDelimiter\abs{\lvert}{\rvert}%
\DeclareMathOperator{\dis}{dis}
\DeclareMathOperator*{\diam}{diam}
\DeclareMathOperator*{\corr}{Corr}
\DeclareMathOperator*{\argmin}{arg\,min}

\theoremstyle{plain}
\title{Temporal Hierarchical Clustering}
\author{
Tamal K. Dey\thanks{
    Dept.~of Computer Science and Engineering, Dept.~of Mathematics, The Ohio State University. Columbus, Ohio.
    \texttt{dey.8@osu.edu}
}
\and
Alfred Rossi\thanks{
    Dept.~of Computer Science and Engineering, The Ohio State University. Columbus, Ohio.
    \texttt{rossi.49@osu.edu}
}
\and
Anastasios Sidiropoulos\thanks{
    Dept.~of Computer Science and Engineering, Dept.~of Mathematics, The Ohio State University. Columbus, Ohio.
    \texttt{sidiropoulos.1@osu.edu}
}
}

\date{}
\begin{document}
\maketitle
\begin{abstract}
We study hierarchical clusterings of metric spaces that change over time. This
is a natural geometric primitive for the analysis of dynamic data sets.
Specifically, we introduce and study the problem of finding a temporally
coherent sequence of hierarchical clusterings from a sequence of unlabeled point
sets. We encode the clustering objective by embedding each point set into an
ultrametric space, which naturally induces a hierarchical clustering of the set
of points. We enforce temporal coherence among the embeddings by finding
correspondences between successive pairs of ultrametric spaces which exhibit
small distortion in the Gromov-Hausdorff sense. We present both upper and
lower bounds on the approximability of the resulting optimization problems.
\end{abstract}


\section{Introduction}
Clustering is a primitive in data analysis which simultaneously serves to summarize data and elucidate its hidden structure.
In its most common form a clustering problem consists of a pair $(P,k)$, where
$P$ is a metric space, and $k$ indicates the desired number of clusters.
The goal of the problem is to try to find a partition of the points of $P$ into $k$ sets such that some objective is minimized.
Because of the fundamental nature of such a primitive, clustering enjoys broad application in a variety of settings and an extensive body of work exists to explain, refine, and adapt its methodology~\cite{DBLP:conf/soda/ArthurV07,DBLP:journals/corr/DeyRS14,DBLP:conf/colt/EldridgeBW15,forgy1965cluster,DBLP:journals/mor/HochbaumS85,DBLP:books/ph/JainD88}.

Having to decide the number of clusters in advance can be a source of difficulty in practice.
When faced with this problem, one common approach is to use hierarchical clustering to produce a parameter free summary of the input. That is, instead of producing a single partition of the input points, the goal is to find a rooted tree (called a \emph{dendrogram}) where the leaves are the points of $P$ and the internal nodes of the tree indicate the distance at which its subtrees merge.

We aim to address the analogous question of how to avoid having to decide the number of clusters in advance in the case of dynamic data. Here, we adopt the temporal clustering framework of
~\cite{DBLP:conf/esa/DeyRS17,DBLP:journals/corr/DeyRS17}. In this framework, the
input is a sequence of clustering problems, and the goal is to ensure that the
solutions of successive instances remain close according to some objective. This
differs from
incremental~\cite{DBLP:conf/nips/AckermanD14,DBLP:conf/stoc/CharikarCFM97} and
kinetic
clustering~\cite{DBLP:conf/compgeom/AbamB09,DBLP:conf/soda/BaschGH97,DBLP:journals/comgeo/FriedlerM10,DBLP:conf/compgeom/GaoGN04}
in that there is no constraint that the clustering instances in the input must
be incrementally related. Further, an optimal sequence of spatial clusterings is
not automatically a low cost solution to the temporal clustering instance.

In this paper we present a natural adaptation of hierarchical clustering to the
temporal setting. We study the problem of finding a temporally coherent sequence
of hierarchical clusterings from a sequence of unlabeled point sets. Our goal is
to produce a sequence of hierarchical clusterings (dendrograms) corresponding to
each set of points in the input such that successive pairs of clusterings have
similar dendrograms. We show that the corresponding optimization problem is
$\NP$-hard. However, a polynomial-time approximation algorithm exists when the
metric spaces in the input are taken from a common ambient metric space. We
explore the properties of this algorithm and find that it is unstable under
perturbations of the metric. We then show how to restore stability with only a
slight loss in the guarantee.

\subsection*{Problem formulation}

An idea used in this paper is that we may
hierarchically cluster a metric space by trying to find a low distortion
embedding of it into an ultrametric. An \emph{ultrametric} is a metric space
which satisfies a stronger version of the triangle inequality. Formally, an
ultrametric space is a metric space $U = (X, \mu)$ such that
$
  \mu(x,z) \leq  \max\{\mu(x,y),\, \mu(y,z)\},
$
for all $x,y,z \in X$.

Ultrametric spaces have interesting geometry. For instance, in an ultrametric
all points contained in a ball of radius $r$ are \emph{centers} of the ball.
That is, for any $q \in \ball_U(p;r)$, we have $\ball_U(q;r) =
\ball_U(p;r)$, where $\ball_M(p; r)$ denotes the ball of
radius $r$ about a point $p$ in a metric space $M$. Further, given any pair of balls
$B \subseteq U$, $\pr{B} \subseteq U$ with non-empty intersection, one has $B
\subseteq \pr{B}$ or $\pr{B} \subseteq B$. This simple fact implies that any
ultrametric space has the structure of a tree where items in a common subtree are
close. That is, an ultrametric induces a natural hierarchical clustering,
commonly depicted as a dendrogram (see Figure~\ref{fig:dendrogram}).

\begin{figure}
\centering
\includegraphics[width=0.6\linewidth]{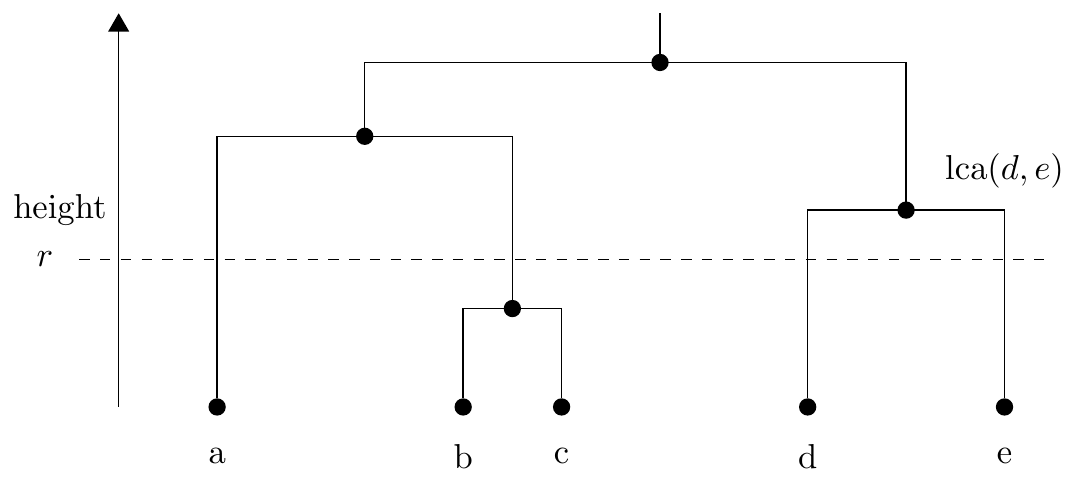}
\caption{\label{fig:dendrogram}
The dendrogram of an ultrametric, $(U, \mu)$, on points $\{a,b,c,d,e\}$. The
points of $\mu$ are the leaves of the dendrogram (height $0$). The distance
between two points $x, y \in U$ is given by the height of their lowest common
ancestor, $\lca(x,y)$, that is at $\mu(x,y)$. The dashed cut at $r$ induces a
natural clustering $\{\{a\}, \{b,c\}, \{d,e\} \}$ of the points of $U$ by
grouping points which belong to the same subtree. Each of these groups are
contained in disjoint balls of radius $r$.}
\end{figure}

\paragraph*{Similarity of dendrograms.}
For dendrograms over sets of points with identical labelings there is a natural
dissimilarity measure given by comparing the merge heights for any pair of
corresponding points. Namely,
$
  \max_{u, \pr{u} \in P} \abs{h_1(u,\pr{u}) - h_2(u,\pr{u})},
$
where $h_1$, and $h_2$ give the merge heights for a respective pair of
dendrograms.

One immediate obstacle to adopting this formalization is that our model does not
require that the sets of points comprising the input have the same cardinality.
For this reason, we take the point of view that two dendrograms are similar if
there exists a \emph{correspondence} between their leaves such that the merge
heights of corresponding points are close. Formally, a correspondence between
$U$ and $V$ is a relation $\mathcal{C} \subseteq U \times V$ such that
$\pi_{U}(\mathcal{C}) = U$, $\pi_{V}(\mathcal{C}) = V$. Here, $\pi_{U}$,
$\pi_{V}$ denote the canonical projections of $U \times V$ to $U$ and $V$
(respectively). Further we use the notation $\corr(U,V)$ to denote the set of
correspondences between $U$, $V$. Given a correspondence $\mathcal{C}$ between
two sets of points $P_1$, $P_2$, we have the following dissimilarity measure
which accounts for differences in the merge heights of a pair of
dendrograms under a correspondence
$\mathcal{C}$. This measure is called the \emph{distortion}~\cite{Bur01}, or the \emph{merge distortion distance} with respect to $\mathcal{C}$~\cite{DBLP:conf/colt/EldridgeBW15}, and is given by
$
  \dis(h_1, h_2; \mathcal{C})
    := \max_{(u,v), (\pr{u}, \pr{v}) \in \mathcal{C}}
       \abs{h_1(u,\pr{u}) - h_2(v,\pr{v})}.
$

\paragraph*{Generalized version.} Our goal, then, is not only to output a
sequence of hierarchical clusterings corresponding to the point sets of the
input, but also to produce an interstitial sequence of low distortion
correspondences linking successive pairs of dendrograms. We quantify the extent
to which an ultrametric faithfully represents an input metric space under the
$\ell^\infty$ norm. Specifically, let $U=(P, d_U)$, $V=(P, d_V)$ be a pair of
finite pseudometric spaces on the same set of points. We define
$
  L^\infty(U,V) = \max_{p,\pr{p} \in P} \abs{d_U(p,\pr{p}) - d_V(p,\pr{p})}.
$
In other words, a pseudometric space $V$ is a good fit for $U$ (and vice-versa) whenever
$L^\infty(U,V)$ is small.

Let $M := (X, d)$ be a pseudometric space. If for any $u,v,w \in X$, it holds that
$d(u, v) \leq \max \{d(u,w), d(w, v)\}$ then we say that $d$ is a
\emph{pseudo-ultrametric} and $M$ is a \emph{pseudo-ultrametric space}. We now
formally define this general version of the problem.

\begin{definition}[Temporal Hierarchical Clustering (Generalized
Version)]\label{def:htcgen}
Let $\mathcal{M} := \{M_i\}_{i=1}^t$ be a sequence of metric spaces, where for
each $i\in [t]$, $M_i=(P_i,\cdot)$, and let $\chi, \rho \in \mathbb{R}_{\geq 0}$.
The goal of the \textsc{Generalized Temporal Hierarchical Clustering} problem is
to find a sequence of pseudo-ultrametric spaces, $\{U_i := (P_i, \mu_i)\}_{i=1}^t$ and a sequence of
correspondences $\{\mathcal{C}_i\}_{i=1}^{t-1}$, where for each $i \in [t]$, we have $L^\infty(M_i, U_i) \leq \chi$, and for any $i\in [t-1]$, $\mathcal{C}_i \in \corr(P_i, P_{i+1})$ with $\dis(\mu_i, \mu_{i+1}; \mathcal{C}_i) \leq \rho$.
Such a clustering is called a \gclus{\chi}{\rho} of $\mathcal{M}$.
\end{definition}

We show in Section~\ref{sec:gclus} that the \textsc{Generalized Hierarchical
Temporal Clustering} problem is $\NP$-hard.

\paragraph*{Local version.} Absent the ambient metric space, the above notion of
distortion would be sufficient to capture the intuitive idea that consecutive
hierarchical clusterings should be close. However, it is easy to produce
examples where symmetries in the input permit low-distortion correspondences
which are manifestly non-local in the ambient space. Thus it makes sense to
further require that any correspondence be local in the ambient metric. We say
that a correspondence $\mathcal{C}$ is \emph{$\delta$-local} provided that
$
  \max_{(u, v) \in \mathcal{C}} d(u,v) \leq \delta,
$
where $d$ is the distance in the ambient space.

We now formalize this version of the problem. Here, the input $P :=
\{P_i\}_{i=1}^t$, consists of a sequence of unlabeled, finite, non-empty subsets
of a metric space $M$. We call such a sequence a \emph{temporal-sampling} of $M$
of length $t$, and refer to individual elements of the sequence ($P_i$ for some
$i \in [t]$) as a \emph{level} of $P$ (see
~\cite{DBLP:conf/esa/DeyRS17,DBLP:journals/corr/DeyRS17}). The \emph{size} of $P$ is simply the sum of the number of points in each level of $P$, that is $\sum_{i=1}^t |P_i|$. Let $M=(X,d)$ be a
metric space. For any $P \subseteq X$ we use the notation $M[P]$ to denote the
restriction of $M$ to $P$, that is, $M[P] = (P, d\big|_{P})$. We have the
following definition:
\begin{definition}[Temporal Hierarchical Clustering (Local Version)]\label{def:htclocal}
Let $P := \{P_i\}_{i=1}^t$ be a temporal-sampling over a metric space $M=(X, d)$, and let $\chi, \delta \in \mathbb{R}_{\geq 0}$. The goal of the
\textsc{Local Temporal Hierarchical Clustering} problem is to find a sequence of
pseudo-ultrametric spaces, $\{U_i\}_{i=1}^t$, where for each $i\in [t]$,
$U_i=(P_i,\cdot)$, and $L^\infty(M[P_i], U_i) \leq \chi$, together with a
sequence of correspondences $\{\mathcal{C}_i\}_{i=1}^{t-1}$ where for any $i \in
[t-1]$, $\mathcal{C}_i \in \corr(P_i, P_{i+1})$ with $\max_{(u,v) \in
\mathcal{C}_i} d(u,v) \leq \delta$. Such a clustering is called a
\lclus{\chi}{\delta}.
\end{definition}

While the general version of the problem is $\NP$-hard, the local version is
trivial and can be computed in $O(n^2)$-time by computing a correspondence
minimizing the Hausdorff distance for each pair of successive levels. 
We highlight this problem for expository purposes as well as a prelude to a
labeled version of the problem.

This version of the problem is further of interest in that it can be used to
approximate the general version such that the resulting distortion is bounded in
terms of $\chi$, and $\delta$. We discuss this topic further in
Section~\ref{sec:gclus}.

\paragraph*{Labeled version.}
There are already several drawbacks with previous versions of the problem in
regard to making concrete cluster assignments. In particular it is unclear how
to coherently assign cluster labels to points given a correspondence. Moreover,
we must account for the fact that the number of points can vary across levels. Taking
the point of view that a good labeling is one in which labels in successive
levels remain close, we opt to allow points to be given multiple labels. Doing
so affords us additional bookkeeping to help ensure that labelings for near by levels
remain local, even across levels which require relatively few labels.

\begin{figure}
\centering
\includegraphics[width=0.6\linewidth]{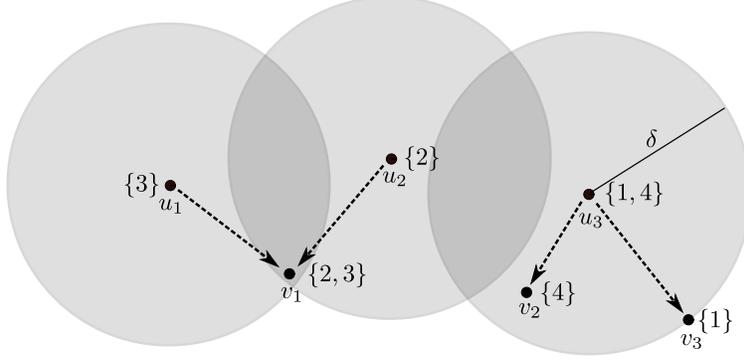}
\caption{\label{fig:labeling}
A $\delta$-contiguous $4$-labeling of $P_1, P_2 \subset \mathbb{R}^2$,
$P_1=\{u_1, u_2, u_3\}$, $P_2=\{v_1, v_2, v_3\}$. Balls of radius $\delta$ are
drawn about the points of $P_1$. Note that the labels used by points of $P_2$
``come from'' points of $P_1$ which are $\delta$-close, demonstrating condition
$2$ of Definition~\ref{def:contiglabel}. The symmetric condition also holds.
Further note that there is no requirement that $|P_1| = |P_2|$.
}
\end{figure}

To this end, given a set $P$, a $k$-labeling of $P$ is a function $L:P
\rightarrow 2^{[k]}$ such that $\{L(p) : p\in P\}$ is a partition of $[k]$.
Informally, we say two labelings are $\delta$-contiguous if the copies of the
same label in a pair of assignments are no farther than $\delta$. We have the
following definition:
\begin{definition}\label{def:contiglabel}
Given a pair of sets $P_1$, $P_2$ of points from a metric space $M$, and a pair
$k$-labelings $L_1$, $L_2$ of $P_1$, $P_2$ (respectively), we say that $L_1$ and
$L_2$ are $\delta$-contiguous in $M$ if
\begin{enumerate}
  \item for all $u \in P_1$, $L_{1}(u) \subseteq \bigcup_{v \in \ball_M(u, \delta) \cap P_2} L_{2}(v)$,
  \item for all $v \in P_2$, $L_{2}(v) \subseteq \bigcup_{u \in \ball_M(v, \delta) \cap P_1} L_{1}(u)$.
\end{enumerate}
See Figure~\ref{fig:labeling} for an example.
\end{definition}

Since points can be multi-labeled, we need a tie-breaking rule to determine
which label applies. By convention we take the label of any set of points to be
the smallest label among all labels of points in the set. Moreover, a good
solution should never use more than $n$ labels on an input of size $n$.
We are now ready to define the main version of the problem.

\begin{definition}[Temporal Hierarchical Clustering]\label{def:htclabeled}
Let $P := \{P_i\}_{i=1}^t$ be a temporal-sampling of size $n$ over a metric
space $M$ with distance, $d$, and let $\chi, \delta \in \mathbb{R}_{\geq 0}$.
The goal of the \textsc{Temporal Hierarchical Clustering} problem is to find a
sequence of pseudo-ultrametric spaces, $\{U_i\}_{i=1}^t$, such that for any $i
\in [t]$, $L^\infty(M[P_i], U_i) \leq \chi$, and a sequence of $k$-labelings,
$\{L_i\}_{i=1}^t$, for $k \leq n$, such that for any $i \in [t-1]$, $L_i$,
$L_{i+1}$ are $\delta$-contiguous. Such a clustering is called a
\clus{\chi}{\delta}.
\end{definition}

\paragraph*{Overview.} In Section~\ref{sec:lclus} we show how to find an
optimal solution to the local version of the problem in $O(n^2)$-time. Then, in
Section~\ref{sec:clus}, we give an $O(n^3)$-time algorithm which converts any
\lclus{\chi}{\delta} into a \clus{\chi}{\delta}. This combined with
Section~\ref{sec:lclus} implies an optimal solution for the labeled version of
the problem. In Section~\ref{sec:gclus} we show that the general version is
$\NP$-hard, but observe that the local version provides an approximate solution
in the special case where the inputs comes from a common metric space. In
Section~\ref{sec:stab} we show that the optimal algorithms are unstable with
respect to perturbations of the metric, and how to ensure stability by
changing the ultrametric construction. Last, Section~\ref{sec:exp} contains an
experiment.

\section{Local Version}\label{sec:lclus}
In this section we present a straightforward solution to the local version of
temporal hierarchical clustering in $O(n^2)$-time. We are not directly
interested in the solution of this problem. Instead, this section serves as a
prelude to solving the labeled version.

\paragraph*{Algorithm.} The algorithm is trivial. Let $\mathcal{A}$ be a
scheme for finding the $\ell^\infty$-nearest ultrametric to a metric. For each
set of points in the input we use $\mathcal{A}$ to find an ultrametric. To
compute correspondences between successive levels $P_i$, $P_{i+1}$, we add all
pairs of points $(u,v) \in P_i \times P_{i+1}$ such that $u$ and $v$ are at a distance of at most the Hausdorff
distance of $P_i$, $P_{i+1}$. Formally, the algorithm takes a temporal-sampling
$P=\{P_i\}_{i=1}^t$ of a metric space $M$ as input and consists of the following
steps:
\begin{description}[itemsep=0pt]
  \item{\textbf{Step 1: Fitting by ultrametrics.}} For each $i \in [t]$, find an
  ultrametric $U_i=\mathcal{A}(M[P_i])$ near to $M[P_i]$ via a chosen scheme.

  \item{\textbf{Step 2: Build correspondences.}} For each $i \in [t-1]$, compute \\
  $\mathcal{C}_i = \{(u,v) \in P_i \times P_{i+1} : d(u,v) \leq d^M_H(P_i, P_{i+1}) \}$.
  Here, $d^M_H$ denotes the Hausdorff distance in the ambient metric space.

  \item{\textbf{Step 3: Return $\left(\{U_i\}_{i=1}^t, \{\mathcal{C}_i\}_{i=1}^{t-1}\right)$.}}
\end{description}

\paragraph*{Analysis.} Let $n$ denote the size of the temporal sampling.
In this section we argue that the above algorithm
returns an optimal solution in $O(n^2)$ time, provided that it is equipped
with a scheme for finding the $\ell^\infty$-nearest ultrametric to an $n$-point
metric space in $O(n^2)$-time. The following theorem ensures that one exists.

\begin{theorem}[Farach-Colton Kannan Warnow~\cite{DBLP:conf/stoc/FarachKW93}]\label{thm:fkw}
Let $M$ be an $n$-point metric space and let $\mathcal{U}(M)$ denote the set of
ultrametrics on the points of $M$. There exists an $O(n^2)$-time algorithm
which finds $\argmin_{U \in \mathcal{U}(M)} L^\infty(U, M)$.
\end{theorem}

We are now ready to prove the main theorem of this section.

\begin{theorem}\label{thm:local}
Let $P$ be a temporal-sampling of size $n$ which admits a \lclus{\chi}{\delta}.
There exists an $O(n^2)$-time algorithm returning a \lclus{\chi}{\delta}.
\end{theorem}
\begin{proof}
Let $t$ denote the length of $P$, and $M$ the ambient metric space. Run the
algorithm of Section~\ref{sec:lclus} where $\mathcal{A}$ is the algorithm of
Farach-Colton, Kannan, and Warnow~\cite{DBLP:conf/stoc/FarachKW93}. Let
$\{U_i\}_{i=1}^t$ denote the pseudo-ultrametrics in the output. By
Theorem~\ref{thm:fkw}, $\pr\chi = \max_{i \in [t]} L^\infty(U_i, M[P_i]) \leq
\chi$, as otherwise $\chi > \pr\chi$ would imply that for some level $i \in
[t]$, the algorithm of Theorem~\ref{thm:fkw} fails to return an
$\ell^\infty$-nearest ultrametric to $P_i$.

Let $\pr\delta=\max_{i\in[t-1]} d^M_H(P_i, P_{i+1})$.
We now argue that $\pr\delta$ is smallest possible in the sense that $P$ admits a \lclus{\chi'}{\delta'}, but does not admit an \lclus{\chi}{\delta} for any $\chi$, when $\delta < \delta'$.
Let $\Gamma := \{\delta : P~\textrm{admits a}~\lclus{\cdot}{\delta}\}$.
First we show $\pr\delta\leq\inf\Gamma$.
Fix any \lclus{\cdot}{\delta}, and let
$\{\mathcal{C}_i\}_{i=1}^{t-1}$ be the associated sequence of $\delta$-local
correspondences. Fix some $1 \leq i < t$ and some $p \in P_i$. Since
$\mathcal{C}_i$ is a correspondence, $\pi_{P_i}(\mathcal{C}_i) = P_i$, and thus
there exists $q \in P_{i+1}$ such that $(p, q) \in \mathcal{C}_i$. Since
$\mathcal{C}_i$ is $\delta$-local it holds that $d(p,q) \leq \delta$, and we
conclude $d(p, P_{i+1}) \leq \delta$. An analogous argument for $q \in P_{i+1}$
implies $d(P_i, q) \leq \delta$. Thus, for $1 \leq i < t$, $
    \pr\delta \leq d^M_H(P_i, P_{i+1}) = \max\left(
        \max_{p \in P_i} d(p, P_{i+1}),
        \max_{q \in P_{i+1}} d(P_i, q)
      \right) \leq \delta.
$
Now we argue that $\pr\delta$ is feasible.
Fix $1 \leq i < t$. Since $d_H(P_i, P_{i+1}) \leq \pr\delta$ it holds that for
every point $p \in P_i$ there exists $q_p \in P_{i+1}$ such that $d(p,q_p) \leq
\pr\delta$. Construct a set $\mathcal{C}_i^+ = \{(p, q_p) : p \in P_i, q_p \in
P_{i+1}, ~\textrm{and}~ d(p, q_p) \leq \pr\delta\}$. Analogously construct a set
$\mathcal{C}_i^- = \{(p_q, q) : q \in P_{i+1}, p_q \in P_{i}, ~\textrm{and}~
d(p_q, q) \leq \pr\delta\}$. The set $\mathcal{C}_i := \mathcal{C}_i^+ \cup
\mathcal{C}_i^-$ is thus a $\pr\delta$-local correspondence between $P_i$,
$P_{i+1}$. Thus, it follows that $\pr\delta \in \Gamma$.

The preceding two paragraphs show that the result is a \lclus{\chi}{\delta}. It
only remains to show the algorithm runs in $O(n^2)$-time. Let $n_i = |P_i|$ for
$i \in [t]$. Step $1$ takes $O(n^2)$-time as finding the $\ell^\infty$-nearest
ultrametric for level $i$ can be done in $O(n_i^2)$-time by
Theorem~\ref{thm:fkw}. Computing the inter-level Hausdorff distance and building
the correspondence for level $i$ in Step $2$ can both be done in
$O(n_i^2)$-time, for a total of $O(n^2)$-time over all.
\end{proof}

\section{Labeled Version}\label{sec:clus}

In this section we show how to convert a \lclus{\chi}{\delta} into a
\clus{\chi}{\delta} in $O(n^3)$-time by transforming a sequence of
$\delta$-local correspondences into a sequence of pairwise $\delta$-contiguous
labelings.

\paragraph*{Network flow.} Drawing upon an idea in
\cite{DBLP:conf/esa/DeyRS17,DBLP:journals/corr/DeyRS17}, we employ minimum cost
feasible flow to find a $\delta$-contiguous labeling with few labels. Formally,
we construct the flow instance as follows: Let $P=\{P_i\}_{i=1}^t$ be a
temporal-sampling. Given the $\delta$-local correspondences of a
\lclus{\cdot}{\delta}, $\{\mathcal{C}_i\}_{i=1}^{t-1}$, the following
construction transforms $P$ into a flow network, $F :=
F(\{\mathcal{C}_i\}_{i=1}^{t-1})$, such that corresponding points in successive
levels are connected by a directed edge which points to the higher indexed
level. Moreover, a source, $s$, connects to each of the points in the first
level, while the sink $\pr{s}$ is the target of a directed edge from each point
in $P_t$. Formally, let $V_i(P) = \{(i,v) : v \in P_i\}$. For $i \in [t-1]$, let
$E_i(P) \subseteq V_i(P) \times V_{i+1}(P)$ such that $\left((i,u), (i+1,
v)\right) \in E_i(P)$ if and only if $(u, v) \in \mathcal{C}_i$. The vertices of
$F$ consist of $s$, $\pr{s}$, and the contents of $V_1(P), \ldots, V_t(P)$. The
edges of $F$ consist of the union of $\{s\} \times V_1(P)$, $V_t(P) \times \{\pr{s}\}$,
and $\bigcup_{i=1}^{t-1} E_i(P)$. Specifically, we seek an integral flow with minimum flow value
such that the in-flow of each vertex of $\bigcup_{i=1}^t V_i(P)$ is at least one.

\paragraph*{Algorithm.} The main idea is to view each correspondence as a
bipartite graph. We concatenate the sequence of correspondences together by
merging overlapping vertices. This allows us to interpret the sequence of
correspondences as a graph. Our goal is then to decompose this graph into a path
cover of small size, which we do by solving a flow instance. Since this graph
only contains edges between points which are close, the resulting labeling will
be contiguous. Formally, we perform the following steps:
\begin{description}[itemsep=0pt]
  \item{\textbf{Step 1: Constructing a flow instance.}} Given a sequence of
  $\delta$-local correspondences of a \lclus{\cdot}{\delta}, construct the
  minimum flow instance $F := F(\{\mathcal{C}_i\}_{i=1}^{t-1})$ as defined
  above.
  \item{\textbf{Step 2: Solve the flow instance.}} Find a minimum cost
  integral flow $f$ in $F$.

  \item{\textbf{Step 3: Decompose the flow.}} Greedily extract unit flows from
  $f$ to construct a list of paths $\{\tau_i\}_{i=1}^k$.
  \item{\textbf{Step 4: Construct label functions.}} Build label functions $L_1,
  \ldots, L_t$ by initializing each to the empty set. Next, for each $\tau_j \in
  \{\tau_i\}_{i=1}^k$, denote $\tau_j$ as the $t$ point sequence $p_1, \ldots, p_t$. Append label $j$ to
  $L_1(p_1), \ldots, L_t(p_t)$.
  \item{\textbf{Step 5: Output.}} Return the labelings $L_1, \ldots, L_t$.
\end{description}

\paragraph*{Analysis.} In this section we show that the above algorithm finds an optimal solution in $O(n^3)$-time on temporal samplings of size $n$.
To this end we now argue that the above network flow instance is feasible.

\begin{lemma}\label{lem:flow}
Let $P=\{P_i\}_{i=1}^t$ be a temporal-sampling. Given the $\delta$-local
correspondences of a \lclus{\cdot}{\delta}, $\{\mathcal{C}_i\}_{i=1}^{t-1}$, the
flow instance $F := F(\{\mathcal{C}_i\}_{i=1}^{t-1})$ is feasible with
value at most $n$.
\end{lemma}
\begin{proof}
For any $1 \leq i \leq t$, any point $p \in P_i$ can be extended to a path from
$P_1$ to $P_t$, by iteratively extending the ends of the path via the
correspondences. Construct a feasible flow $f$ by initializing $f$ to be zero
everywhere. Greedily extend points receiving no flow to paths from $P_1$ to
$P_t$ in the described manner, and increase the flow value of $f$ along the path
by $1$. It follows that $f$ remains integral and satisfies all lower bounds of
$F$. Since we flow at most $1$ unit of flow per point of $P$, the value of $f$
is at most $n$.
\end{proof}

The next theorem shows that the algorithm outputs an optimal clustering.

\begin{theorem}
Let $P$ be a temporal-sampling of size $n$. There exists an $O(n^3)$-time
algorithm which is guaranteed to output a \clus{\chi}{\delta} of $P$,
for any $\chi$, $\delta$ such that $P$ admits a \clus{\chi}{\delta}.
\end{theorem}

\begin{proof}
Let $t$ be the length of $P$. Run the algorithm of Section~\ref{sec:lclus} on
$P$. Since $P$ admits a \clus{\chi}{\delta}, it also admits a
\lclus{\chi}{\delta} where for any $1 \leq i < t$, the $i$-th correspondence is
given by
$
  \mathcal{C}_i=\{(u, v) : (u,v) \in P_i \times P_{i+1},\, L_i(u) \cap L_{i+1}(v) \neq \varnothing\}.
$
Thus, by Theorem~\ref{thm:local}, we are guaranteed a \lclus{\chi}{\delta} in
$O(n^2)$-time. Let $\{\mathcal{C}_i\}_{i=1}^{t-1}$ be its $\delta$-local
correspondences, and run the above algorithm on it. By Lemma~\ref{lem:flow}, the
flow instance $F := F(\{\mathcal{C}_i\}_{i=1}^{t-1})$ is feasible with value at
most $n$. Using an algorithm of Gabow \& Tarjan \cite{DBLP:journals/siamcomp/GabowT89}, we can solve $F$ in
$O(n^3)$-time, yielding an integral flow $f$. Again in $O(n^3)$-time, we
decompose $f$ into a collection of unit flows $\{\tau_j\}_{j=1}^k$, for some $k
\leq n$, which we interpret as paths from $P_1$ to $P_t$. 

We now verify that the sequence of label functions output by the algorithm is
indeed a $\delta$-contiguous $k$-labeling for some $k \leq n$. For any $i \in
[t]$, and any $j \in [k]$ let $\tau_j(i)$ denote the $i$-th vertex in the $j$-th
path. Recall that for each $i \in [t]$, we assign each point $u \in P_i$ the set
of labels $L_i(u) = \{j : j \in [k],\, u = \tau_j(i)\}$. Note that each label in
$[k]$ is used at most once per level since for any $j$, $i \in [t]$, $\tau_j(i)$
is the only place where $\tau_j$ intersects $P_i$. Also, since each $\tau_j$
intersects all levels $i \in [t]$, each label is used at least once per level.
It follows that $\{L_i(u): u \in P_i\}$ is a partition of $[k]$. Finally, since
the edges of $F$ correspond to points that are separated by at most $\delta$ in
the ambient space, any two uses of the label $j \in [k]$ for some $i \in [t-1]$
occur within $d(\tau_j(i), \tau_j(i+1)) \leq \delta$. Thus the corresponding
sequence of $k$-labelings is indeed pairwise $\delta$-contiguous.
\end{proof}

\section{Generalized Version}\label{sec:gclus}
In this section we show that the generalized version problem is $\NP$-hard.
However, we argue that for the special case where the points of the input share a
(known) common ambient metric, the algorithm of Section~\ref{sec:lclus} gives an
approximate solution. It remains an open question as to how to find an
approximate solution in polynomial-time when there is no ambient metric (or it is
unknown).

\paragraph*{\NP-hardness.}
Let $G=(V,E)$ be an instance of $3$-coloring. We construct an instance of
\textsc{Generalized Temporal Hierarchical Clustering}, $\mathcal{M}(G)$,
consisting of two levels. For the first level let $P=\{r, g, b\}$ be a set of
three points, and let $d_P$ be a metric on $P$ such that distinct $p, \pr{p} \in
P$ have $d_P(p,\pr{p})=2$. Denote the corresponding metric space $M_P := (P,
d_P)$. For the second level we construct a metric space $M_V := (V, d_V)$, where
$d_V : V \times V \rightarrow
\mathbb{R}_{\geq 0}$, such that
\[
  d_V(u,v) =
    \begin{cases}
      2 & \textrm{if}~\{u,v\} \in E \\
      1 & \textrm{if}~\{u,v\} \not\in E ~\textrm{and}~u\neq v, \\
      0 & \textrm{otherwise}.
    \end{cases}
\]

\begin{lemma}\label{lem:3colorable-cheap}
If $G$ admits a $3$-coloring requiring $3$ colors, then $\mathcal{M}(G)$
admits a \gclus{1}{0}.
\end{lemma}

\begin{proof}
Fix a $3$-coloring of $G=(V,E)$. We will exhibit a pair of pseudometric spaces
and a $0$-distortion correspondence between them. For the first space let $U_P =
(P, \mu_P)$ be a uniform metric space where distinct points are at a distance of
$1$. Note that $L^\infty(M_P, U_P) = \max_{u, \pr{u} \in P} \abs{d_P(u, \pr{u})
- \mu_P(u, \pr{u})}=1$, since for any distinct $u, \pr{u} \in P$, $\abs{d_P(u,
\pr{u}) - \mu_P(u, \pr{u})} = \abs{2 - 1} = 1$.

We will use the points of $P$ to denote the color class of $v \in V$. Fix $c : V
\rightarrow P$ be such that $c(v) = c(\pr{v})$ if and only if $v, \pr{v}$ share
the same color class. Let $U_V = (V, \mu_V)$ be the pseudometric space where for
any $v, \pr{v} \in V$, $\mu_V(v, \pr{v}) = 1$ if and only if $c(v) \neq
c(\pr{v})$, and $\mu_V(v, \pr{v}) = 0$ otherwise. We now bound $L^\infty(M_V,
U_V)$ by considering $\abs{d_V(v, \pr{v}) - \mu_V(v, \pr{v})}$ for an arbitrary
pair $v, \pr{v} \in V$. Since $\mu_V(v, v) = d_V(v, v) = 0$ for any $v \in V$,
only distinct $v, \pr{v}$ can contribute to the distortion. Suppose $\{v,
\pr{v}\} \in E$, then $c(v) \neq c(\pr{v})$ and thus $\abs{d_V(v, \pr{v}) -
\mu_V(v, \pr{v})} = \abs{2 - 1} = 1$. Otherwise, $\{v, \pr{v}\} \not \in E$, and
$d_V(v, \pr{v}) = 1$ while $\mu_V(v, \pr{v}) \leq 1$ so that $\abs{d_V(v,
\pr{v}) - \mu_V(v, \pr{v})} \leq 1$. Thus $L^\infty(M_V, U_V) \leq 1$.

Last, let $\mathcal{C} = \{(p, v) \in P \times V : c(v) = p \}$. We now verify
that $\mathcal{C}$ is a $0$-distortion correspondence. To see that $\mathcal{C}
\in \corr(P, V)$, note that $\pi_P(\mathcal{C}) = P$ since $G$ requires $3$
colors, and $\pi_{V}(\mathcal{C}) = V$ since every vertex $v \in V$ belongs to a
color class. Finally, to bound $\dis(\mu_P, \mu_V ; \mathcal{C})$ note that for
any $(p, v),(\pr{p}, \pr{v}) \in \mathcal{C}$, either $p = \pr{p}$ and
$\abs{\mu_P(p,\pr{p}) - \mu_V(v, \pr{v})}=\abs{\mu_V(v, \pr{v})} = 0$ (since
$c(v) = c(\pr{v})$), or $p \neq \pr{p}$ and $\abs{\mu_P(p,\pr{p}) - \mu_V(v,
\pr{v})}=\abs{1-1} = 0$.
\end{proof}

\begin{lemma}\label{lem:no3coloring-expensive}
If $G$ does not admit a $3$-coloring, then $\mathcal{M}(G)$ does not admit
a \gclus{2}{0}.
\end{lemma}

\begin{proof}
Let $(V,E) = G$. Fix a \gclus{\chi}{0} of $\mathcal{M}(G)$ for some $\chi < 2$
consisting of ultrametrics $U_P=(P, \mu_P)$, $U_V=(V, \mu_V)$, and a
$0$-distortion correspondence $\mathcal{C} \in \corr(P, V)$. We first argue that
the points of $P$ are separated. Let $p, \pr{p} \in P$, $p \neq \pr{p}$. If
$\mu_P(p, \pr{p}) = 0$ then $L^\infty(M_P, U_P) \geq \abs{\mu_P(p, \pr{p}) -
d_P(p, \pr{p})} = |0 - 2| = 2$. Thus $\chi \geq 2$, a contradiction.

Now fix a map $c:V\rightarrow P$, such that for any $v \in V$, $c(v) = p$ such
that $(p, v) \in \mathcal{C}$.  First we argue that $c$ is indeed a function
by showing that for any $v \in V$, $v$ corresponds to exactly one point in $P$.
To see why observe that given any $(p, v), (\pr{p}, v) \in \mathcal{C}$ with $p
\neq \pr{p}$ it follows that $0 = \dis(\mu_P, \mu_V; \mathcal{C}) \geq
\abs{\mu_P(p, \pr{p}) - \mu_V(v,v)} = \mu_P(p, \pr{p}) > 0$. We now show how to
use $c$ to construct a $3$-coloring of $G$. Since $\chi < 2$, for every $\{u,v\}
\in E$, we have $\mu_V(u,v) > 0$, as otherwise $\chi \geq L^\infty(M_V, U_V)
\geq |d_V(u,v) - \mu_V(u,v)| = 2$. Consider any pair of corresponding points
$(c(u), u), (c(v), v) \in \mathcal{C}$. It must be the case that $c(u) \neq
c(v)$ as otherwise $\dis(\mu_P, \mu_V; \mathcal{C}) \geq \abs{\mu_P(c(u), c(v))
- \mu_V(u,v)} = \mu_V(u,v) > 0$. Color the graph by assigning each $v \in V$ to
a color class given by $c(v)$. Since for adjacent $u, v \in V$, we have
$\mu_V(u,v) > 0$, it follows that $c(u) \neq c(v)$, and thus there is no edge
between vertices of the same color. We have exhibited a $3$-coloring of $G$.
\end{proof}

Theorem~\ref{thm:hard} result follows directly from
Lemma~\ref{lem:3colorable-cheap}, and Lemma~\ref{lem:no3coloring-expensive}. The
proof also implies that for the \textsc{Generalized Temporal Hierarchical
Clustering} problem, for some fixed $\rho$, approximating $\chi$ within any
factor smaller than $2$ is $NP$-hard.

\begin{theorem}\label{thm:hard}
The \textsc{Generalized Temporal Hierarchical Clustering} problem is $NP$-hard.
\end{theorem}

\paragraph*{Approximation by local version.} We now show that any
\lclus{\chi}{\delta} is a \gclus{\chi}{2\chi + 2\delta}. That is, we can view
the local version of the problem as an approximation to the general version in
the special case that the points of the input come from the same metric space. 

\begin{lemma}\label{lem:bdddist}
Let $P$ be a temporal-sampling. Any \lclus{\chi}{\delta} of $P$ is a \gclus{\chi}{2\chi + 2\delta} of $P$.
\end{lemma}

\begin{proof}
Suppose $P$ has length $t$ and ambient metric space $M=(X,d)$. Fix a \lclus{\chi}{\delta} of $P$ with ultrametrics
$\{U_i=(P_i, \mu_i)\}_{i=1}^t$, and correspondences, $\{\mathcal{C}_i\}_{i=1}^{t-1}$, induced by labelings of successive pairs of levels.
Observe that
\[
  \max_{i \in [t-1]} \dis(\mu_i, \mu_{i+1}, \mathcal{C}_i)
     = \max_{i \in [t-1]} \max_{(x,y),(\pr{x}, \pr{y}) \in \mathcal{C}_i}
       \abs{\mu_i(x,\pr{x}) - \mu_{i+1}(y,\pr{y})}.
\]
Since $\chi \geq \max_{i \in [t]} L^\infty(M[P_i], U_i)$, it follows by definition
of $L^\infty$ that $\chi \geq \abs{\mu_i(x, \pr{x}) - d(x, \pr{x})}$ for any $i
\in [t]$, $x,\pr{x} \in P_i$. Fix an arbitrary $i \in [t-1]$ and let $(x,y),
(\pr{x}, \pr{y}) \in \mathcal{C}_i$. By triangle inequality $\abs{\mu_i(x,\pr{x})
- \mu_{i+1}(y,\pr{y})} \leq \abs{d(x,\pr{x}) - d(y,\pr{y})} + 2\chi$. Note that
since $(x,y), (\pr{x}, \pr{y}) \in \mathcal{C}_i$, we have $d(x,y), d(\pr{x},
\pr{y}) \leq \delta$. Thus $y$, $\pr{y} \in X$ are contained in $\delta$-balls
of $x$, $\pr{x}$ in $X$ (respectively). It follows that $\abs{d(x,\pr{x}) - d(y,
\pr{y})} \leq 2\delta$. We conclude that for any $i \in [t-1]$, $(x,y), (\pr{x},
\pr{y}) \in \mathcal{C}_i$, $\abs{\mu_i(x,\pr{x}) - \mu_{i+1}(y,\pr{y})} \leq 2\chi
+ 2\delta$, and thus $\max_{i \in [t-1]} \dis(\mu_i, \mu_{i+1}; \mathcal{C}_i) \leq
2\chi + 2\delta$.
\end{proof}

\section{Stability}\label{sec:stab}
In this section we show that the algorithm for finding an $\ell^\infty$-nearest
ultrametric in~\cite{DBLP:conf/stoc/FarachKW93} is unstable under perturbations
of the metric and, consequently, so are our algorithms. Stability, naturally, is
a desirable property; as otherwise if small changes in the input are allowed to
produce vastly different ultrametrics, then the observed temporal coherence of
the output is lost. Furthermore, this is the case even if the cost of fitting
each level to an ultrametric remains best possible. We resolve this issue in
practice by instead finding the $\ell^\infty$-nearest subdominant ultrametric.

\paragraph*{Subdominant ultrametrics.} Let $M=(X,d)$ be a metric space. We will consider $M$ to be a complete graph
where the edges are weighted by distance, and use the notation $T_M$ to refer to
a minimum spanning tree on $M$. Further, for any $x, y \in M$, let $T_M(x,y)$
denote the unique path joining $x, y \in M$.
Let $\mathcal{U}(M)$ denote the set of ultrametrics on the points of $M$. Let
$
  \mathcal{U}_\leq(M) = \{(X, \mu) \in \mathcal{U}(M) : \mu(x,y) \leq d(x,y) ~\textrm{for all}~ x,y \in M \}.
$
In other words, $\mathcal{U}_\leq(M)$ is the set of ultrametrics on the points
of $M$ such that no distance is made larger than its counterpart in $M$. We say
that an ultrametric in $\mathcal{U}_\leq(M)$ is \emph{subdominant} to $M$. Let
$\sub{\mu}(M)=(U, \mu)$ be a metric space on the points of $M$ with distance
function
$
  \mu(x,y) = \max_{\{u,v\} \in T_M(x,y)} M(u,v).
$
The distance function $\mu$ is independent of the choice of minimum spanning tree, and
easily verified to be ultrametric and subdominant to $M$. It can further be
shown that $\sub{\mu}(M)$ is the unique, $\ell^\infty$-closest subdominant ultrametric to
$M$. That is,
$
  \mu_S(M) = \argmin_{U \in \mathcal{U}_\leq(M)} L^\infty(U, M).
$

\paragraph*{Instability.} We now show that the algorithms of
Section~\ref{sec:lclus}, Section~\ref{sec:clus} are unstable. To elucidate why
we now restate the algorithm in~\cite{DBLP:conf/stoc/FarachKW93} in a slightly
modified form which helps to make our point. This procedure is equivalent to the
following:
\begin{description}[itemsep=0pt]
  \item{\textbf{Step 1: Compute a minimum spanning tree.}} 
    Given a metric space $M=(X,d)$ consider a weighted complete graph on $X$
    where the the weight of any edge $\{x, \pr{x}\}$ is $d(x,\pr{x})$. Find a
    minimum spanning tree of this graph, $T_M$.

  \item{\textbf{Step 2: Compute cut-weights for each edge.}} 
    Let $(X, \mu) = \mu_S(M)$. For each edge $e=\{u, v\} \in T_M$, compute and assign a priority $p(e)$
    to $e$ such that
    $
      p(e) = \max_{x,\pr{x} \in X} \{d(x,\pr{x}) : e \in T_M(x,\pr{x}),\, \mu(x,\pr{x}) = d(u,v) \},
    $

  \item{\textbf{Step 3: Assign distances.}} 
    Edges are cut in order of descending priority. Any pair of vertices $u,v
    \in T_M$ first separated by a cut at $e$ are assigned a distance of
    $p(e)-\frac{1}{2}L^\infty(M, \mu_S(M))$.

\end{description}
When an edge is cut, points first separated by the removal of that edge are
assigned a distance which depends on its largest supported distance in $M$. The
issue is that small perturbations in the metric can change the path structure of
$T_M$ so that an edge becomes responsible for linking a far pair of points. The
only hope for stability is that the other term in the assigned distance,
$\frac{1}{2}L^\infty(M, \mu_S(M))$, changes enough to offset this effect.
However, Lemma~\ref{lem:subulmstable} shows that this term is stable, and
thus is not large enough to compensate. It follows that the above procedure is
unstable. See Figure~\ref{fig:diagram-unstable-nearest-ulm} for a concrete
example.

\begin{figure}
\centering
\includegraphics[width=1\linewidth]{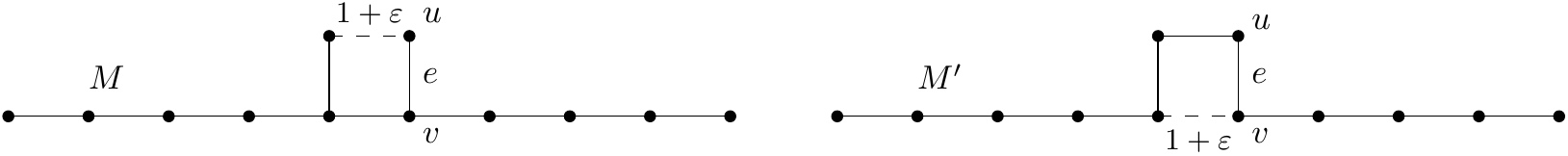}
\caption{\label{fig:diagram-unstable-nearest-ulm}
Two metric graphs $M$, $\pr{M}$ which differ by an $\eps$-perturbation. Solid
edges have length $1$ and appear in their respective MSTs. Dashed edges have
length $1+\eps$. In $M$, $p(e) = 6$, and the priority of any edge along the
bottom of $M$ is $\diam(M)=9$. Let $U$, $\pr{U}$ denote the result from running
the algorithm in~\cite{DBLP:conf/stoc/FarachKW93} on $M$, and $\pr{M}$,
respectively. In computing $U$ from $M$ the edge $e$ is cut after all of the
edges along the base, and thus $u,v$ are assigned distance of
$6-\frac{1}{2}L^\infty(M,\mu_S(M))=6-\frac{1}{2}(9-1) = 2$. Compare with
$\pr{M}$, where the priority $p$ of any edge of the MST is $p = p(e) =
\diam(\pr{M}) = 9+\eps$. Thus $u, v$ are assigned distance of
$9+\eps-\frac{1}{2}L^\infty(\pr{M},\mu_S(\pr{M}))=5+\eps/2$. By considering $n$
point metric spaces with bases of length $n-2$, this example generalizes to show
$L^\infty(U, \pr{U}) = \Omega(\diam(M))$.
}
\end{figure}

\paragraph*{Ensuring stability.}
In contrast, the $\ell^\infty$-nearest subdominant ultrametric is
stable under metric perturbations. We now give a simple, direct proof of this fact for our setting. See~\cite{DBLP:journals/jmlr/CarlssonM10} for extended discussion.

\begin{lemma}\label{lem:subulmstable}
Let $M$, $\pr{M}$ be metric spaces on the same points such that
$L^\infty(M, \pr{M}) \leq \eps$, then $L^\infty(\mu_S(M),\mu_S(\pr{M})) \leq \eps$.
\end{lemma}
\begin{proof}
Let $P$ denote the points of $M$. Fix a distance weighted MST of $M$, $T_M$, and let $(P, \mu) = \mu_S(M)$, $(P, \pr\mu)=\mu_S(\pr{M})$. For any
pair of points $x, y \in P$ let $\mathcal{P}(x,y)$ denote the set of all simple
paths $x \leadsto y$ in $M$ (when $M$ is viewed as a complete graph). Let $w :
\mathcal{P}(x,y) \rightarrow \mathbb{R}_{\geq 0}$ be the function that sends
each path in $\mathcal{P}(x,y)$ to the value of its maximum weight edge. Observe
that the maximum weight edge along $T_M(x,y)$ is equal to $\min_{\gamma \in
\mathcal{P}(x,y)} w(\gamma)$, as otherwise it is possible to construct a
spanning tree with cost strictly less than that of $T_M$. Thus, $\mu(x,y) =
\min_{\gamma \in \mathcal{P}(x,y)} w(\gamma)$. Now since $M$, $\pr{M}$ differ by
an $\eps$-perturbation, the values individual edges of the paths (and therefore
the values of the paths in $\mathcal{P}(x,y)$ under $w$) change by at most
$\eps$. Thus, $|\mu(x,y) - \pr\mu(x,y)| \leq \eps$
\end{proof}

Such a choice for ultrametric embedding is suboptimal, but
the next lemma shows that it is within a factor of $2$ of optimal. This fact essentially follows from arguments
in~\cite{DBLP:conf/stoc/FarachKW93}, but we give a proof by different means.

\begin{lemma}[\cite{DBLP:conf/stoc/FarachKW93}]\label{lem:subdomapprox}
Let $M$ be a finite metric space and $U \in \mathcal{U}(M)$, then
$L^\infty(\sub{\mu}(M), M) \leq 2 L^\infty(U, M)$.
\end{lemma}

\begin{proof}
Let $U^*$ denote an $\ell^\infty$-nearest ultrametric on the points of $M$, and
suppose that $L^\infty(M,U^*) = \eps$. It follows immediately that there exists
an $\eps$-perturbation of $M$, $\overline{\eps}$, such that $M = U^* +
\overline{\eps}$. Now,
$
  \mu_S(M) = \mu_S(U^* + \overline{\eps})
           = \mu_S(U^*) + \pr{\overline{\eps}}
           = U^* + \pr{\overline{\eps}},
$
for some other $\eps$-perturbation $\pr{\overline{\eps}}$. Here, the second
equality follows by stability (Lemma~\ref{lem:subulmstable}), and the third
follows from the fact that $U^*$ is its own $\ell^\infty$-nearest subdominant
ultrametric. Thus,
$
  L^\infty(\mu_S(M), M) 
      = L^\infty(U^* + \pr{\overline{\eps}}, M)
    \leq L^\infty(U^*, M) + \eps
      = 2 L^\infty(U^*, M).
$
The proof follows by noting that = $L^\infty(U^*, M) \leq L^\infty(U, M)$.
\end{proof}

As one might expect, using the $2$-approximate algorithm $\mu_S$ for
$\mathcal{A}$ in the algorithm of Section~\ref{sec:lclus} results in a
\lclus{2\chi}{\delta} whenever the input admits a \lclus{\chi}{\delta}.
Lemma~\ref{lem:bdddist} then implies that the result is a \gclus{2\chi}{4\chi +
2\delta}. However, since the error incurred by $\mu_S$ is one-sided, there is no
additional loss in the coupling distortion and the result
is a \gclus{2\chi}{2\chi + 2\delta}.

\begin{proof}[Proof sketch]
Let $P$ be a temporal-sampling of length $t$ from a metric space $M$ with
metric $d$.
Suppose that $U_i=(P_i, \mu_i)$ is a subdominant pseudo-ultrametric to $M[P_i]$ for all $i \in [t]$. Then, for
any $x, \pr{x} \in P_i$, $0 \leq d(x,\pr{x})-\mu_i(x,\pr{x}) \leq \chi$. That
is, $\mu_i(x,\pr{x}) \in [d(x, \pr{x})-\chi, d(x,\pr{x})]$. Thus, for any $i \in
[t-1]$, $x,\pr{x} \in P_i$, $y,\pr{y} \in P_{i+1}$ it follows that
$\abs{\mu_i(x,\pr{x}) - \mu_{i+1}(y,\pr{y})} \leq \abs{d(x,\pr{x}) - d(y,
\pr{y})} + \chi$. Using this bound in the proof of Lemma~\ref{lem:bdddist}
gives the desired result.
\end{proof}

\section{Example Output}\label{sec:exp}
In Figure~\ref{fig:boids}, we present
output based on synthetic data. For expository purposes we seek a data source
for which many levels can reasonably be described as hierarchical, yet
changes enough that the hierarchy evolves over time. We obtain such input by
regularly saving snapshots of actor positions from a flocking simulation.
A labeled clustering is obtained using the algorithm of Section~\ref{sec:clus}
and fitting by subdominant ultrametrics.

\begin{figure}%
\centering
\begin{subfigure}{0.25\linewidth}
\centering
\includegraphics[width=\linewidth]{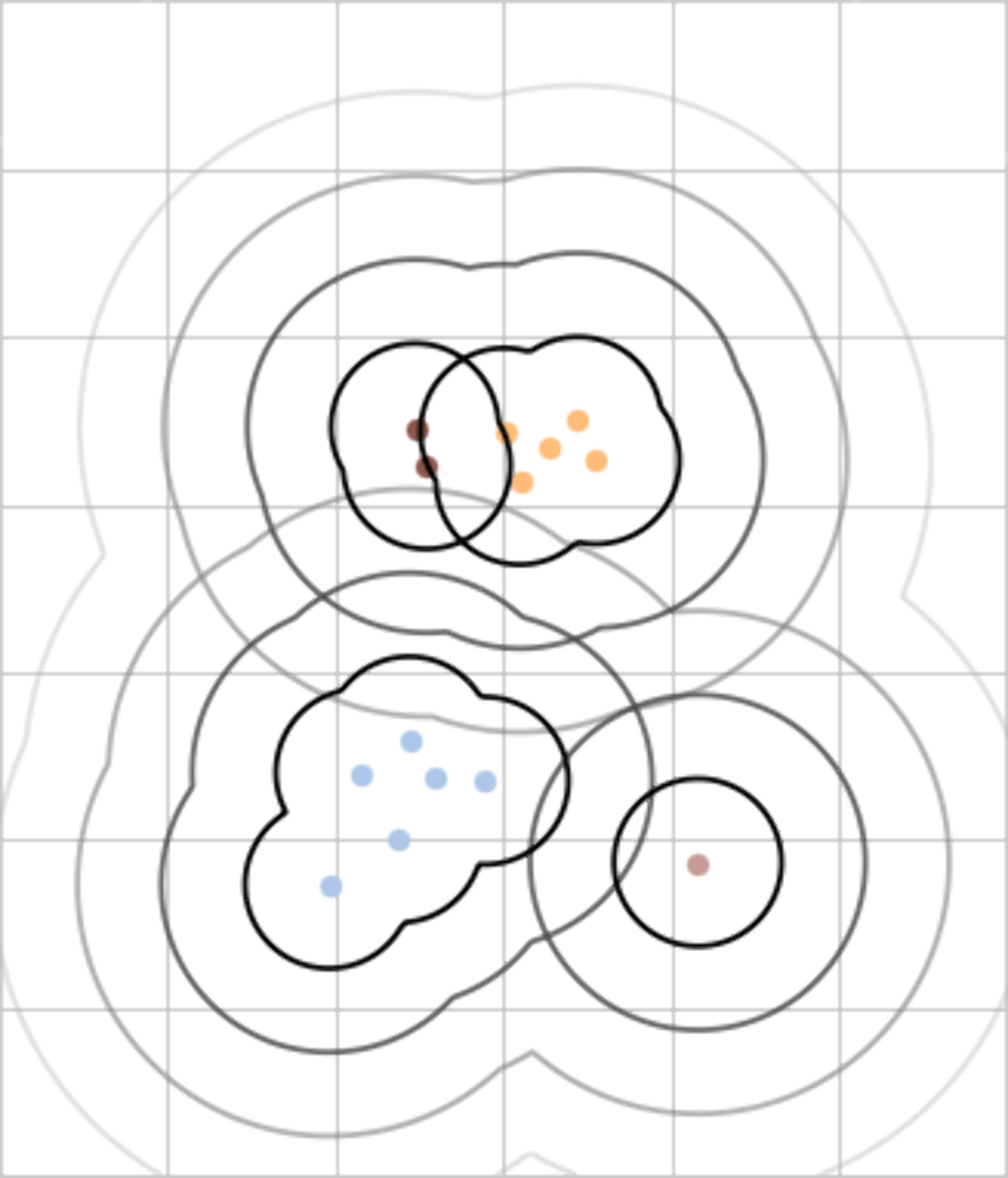}
\caption{}
\label{fig:boids1}
\end{subfigure}
\begin{subfigure}{0.25\linewidth}
\centering
\includegraphics[width=\linewidth]{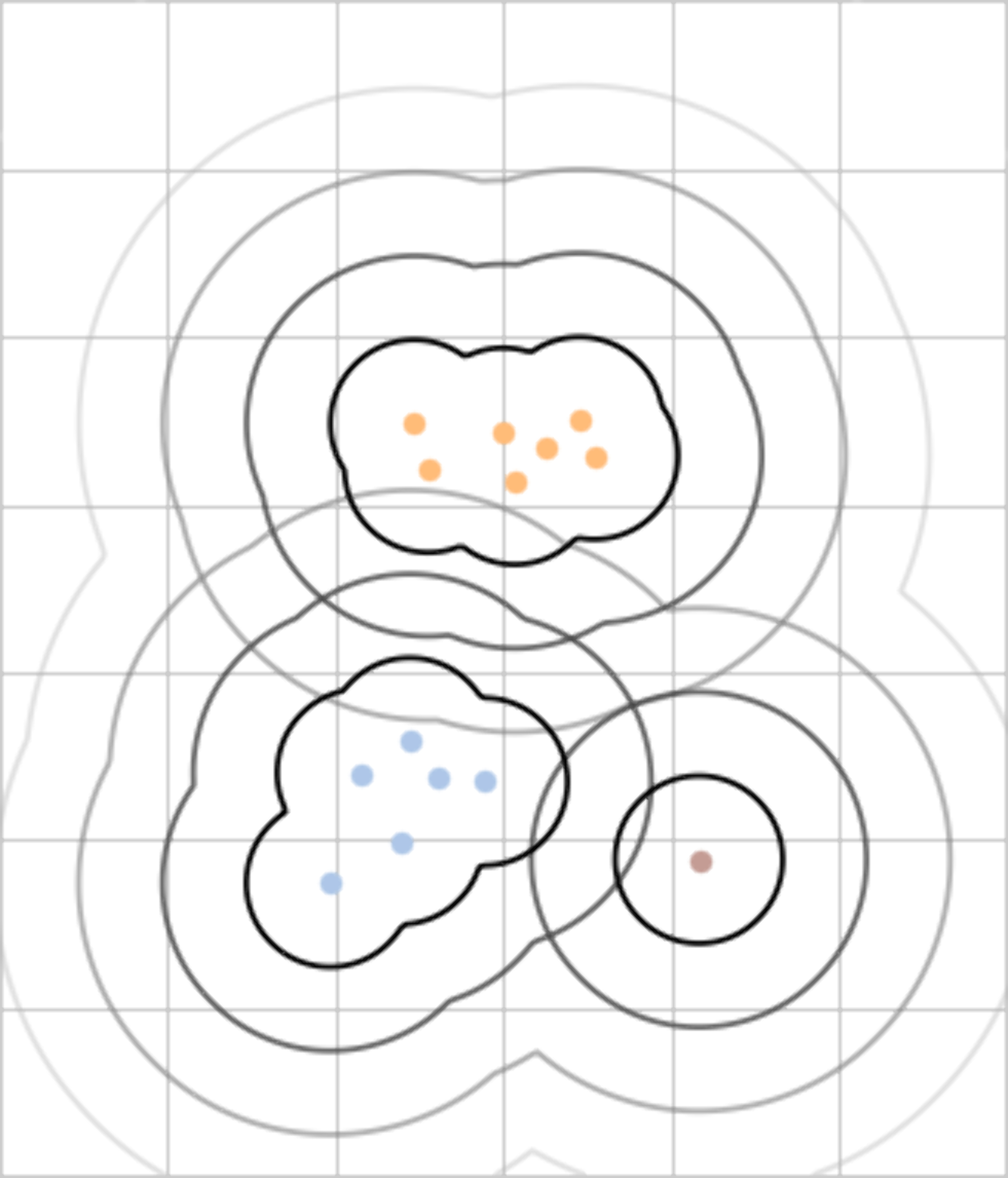}
\caption{}
\label{fig:boids2}
\end{subfigure}
\begin{subfigure}{0.25\linewidth}
\centering
\includegraphics[width=\linewidth]{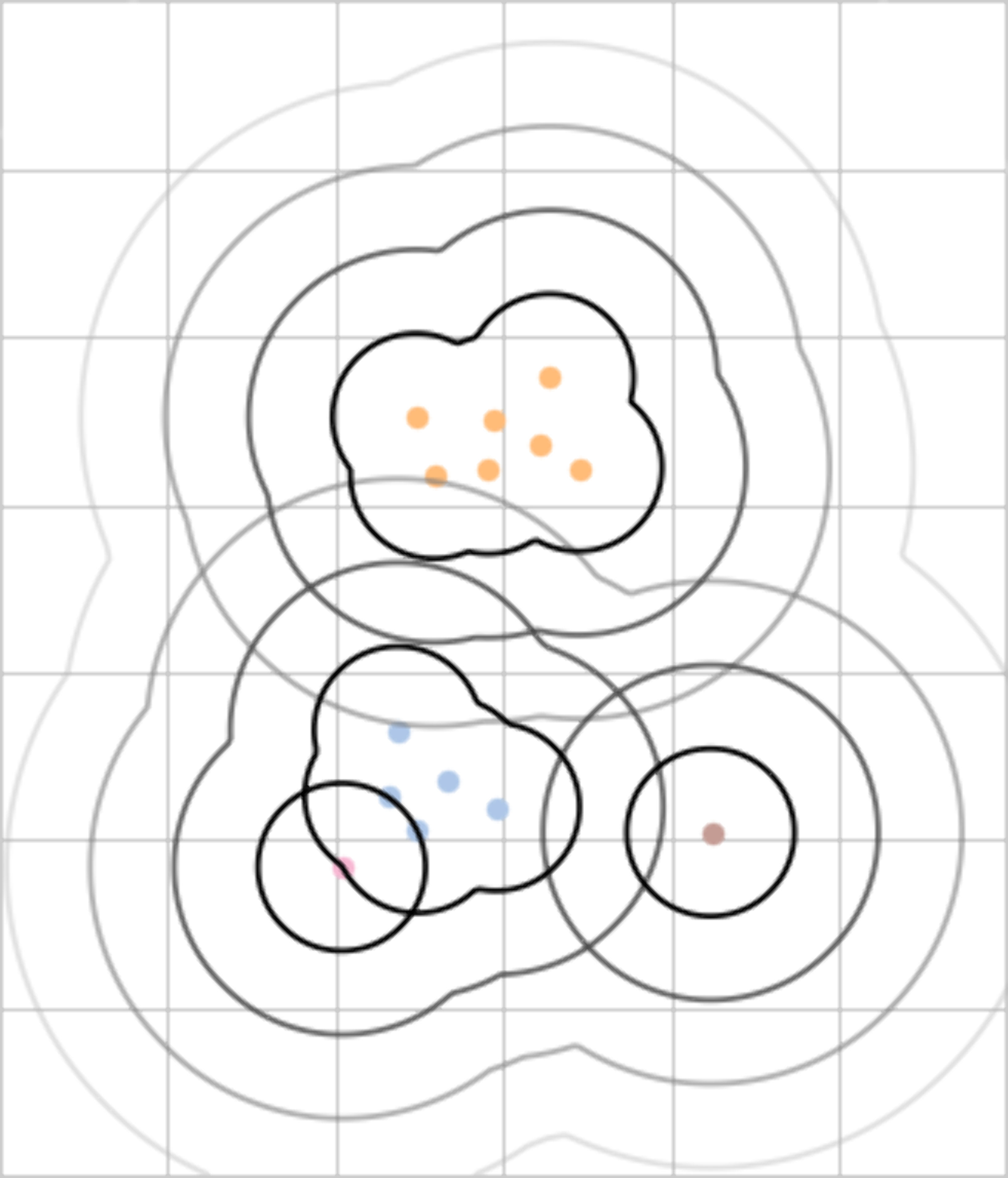}
\caption{}
\label{fig:boids3}
\end{subfigure}
\caption{Three levels of a temporal hierarchical clustering. The
  contours show the coarse cluster structure which results from cutting the
  ultrametric at various offsets. Points which appear together within a contour
  share a cluster at that height in the ultrametric tree.
  \textbf{(\ref{fig:boids1})} Yellow and brown clusters are close.
  \textbf{(\ref{fig:boids2})} One level later, yellow and brown clusters
  merge. Note that the coarse structure remains stable.
  \textbf{(\ref{fig:boids3})} Ten levels later, a blue point (now pink)
  splits from its cluster.
  \label{fig:boids}}%
\end{figure}

For completeness we now describe the rules of the simulation in detail. At
initialization, a fixed number of actors are spawned on the plane. An actor can
be one of four types, assigned uniformly at random at the time of creation. The
simulation `arena' consists of a neutral square area which exerts no force on
the actors, surrounded by repulsive walls which serve to keep the actors in the
neutral area. The actors exhibit \emph{clumping}, \emph{avoidance}, and
\emph{schooling} behaviors. The clumping rule attracts actors to other nearby
actors of the same type. The avoidance rule causes close actors (of any type) to
repel each other. The schooling behavior weakly accelerates actors in the
direction of the global average velocity among actors of the same type.
Since
our objectives can handle changes in the number of points we include the
following behaviors: When actors are very close to each other they may interact
with some small probability. If actors of the same type interact they may
produce a new actor (of random type). If actors of different type interact one
of them may be deleted from the simulation.

It may be tempting to ask whether or not the clustering can be used to recover
actor types from the output clusterings. The answer, unfortunately, is ``no''.
While this can likely be determined at sufficiently high temporal resolution by
examining the apparent velocities of corresponding points, our clustering
procedure is not sensitive to this and forces proximate points to share a
cluster label.

\section*{Conclusion}
We conclude by briefly mentioning some open questions. In
Section~\ref{sec:gclus} we show that the general problem is $\NP$-hard, though
our proof uses an unnatural metric space. It is unknown if the general version
admits an exact algorithm on ``nice'' metric spaces. Further, it may still be
possible to obtain optimal algorithms for the local and labeled versions of the
problem which are stable under metric perturbations. Last, while we believe that
our adaptations of hierarchical clustering are quite natural, one could consider
alternative models where, say, the distortion is replaced with a tree
dissimilarity measure (e.g. nearest neighbor interchange).

\paragraph*{Grant Acknowledgments.}
This work was partially supported by the NSF grants CCF 1318595, CAREER 1453472, CCF 1423230,
and DMS 1547357.

\bibliographystyle{plainurl}
\bibliography{bibfile}

\begin{thebibliography}{10}

\bibitem{DBLP:conf/compgeom/AbamB09}
M.~Ali Abam and M.~de~Berg.
\newblock Kinetic spanners in $\mathbb{R}^d$.
\newblock In {\em {SOCG}}, 2009.
\newblock URL: \url{http://doi.acm.org/10.1145/1542362.1542371}.

\bibitem{DBLP:conf/nips/AckermanD14}
M.~Ackerman and S.~Dasgupta.
\newblock Incremental clustering: The case for extra clusters.
\newblock In {\em {NIPS}}, 2014.

\bibitem{DBLP:conf/soda/ArthurV07}
D.~Arthur and S.~Vassilvitskii.
\newblock k-means++: the advantages of careful seeding.
\newblock In {\em {SODA}}, 2007.
\newblock URL: \url{http://dl.acm.org/citation.cfm?id=1283383.1283494}.

\bibitem{DBLP:conf/soda/BaschGH97}
J.~Basch, L.~J. Guibas, and J.~Hershberger.
\newblock Data structures for mobile data.
\newblock In {\em {SODA}}, 1997.
\newblock URL: \url{http://dl.acm.org/citation.cfm?id=314161.314435}.

\bibitem{Bur01}
D.~Burago, Y.~Burago, and S.~Ivanov.
\newblock {\em A Course in Metric Geometry}.
\newblock AMS, 2001.

\bibitem{DBLP:journals/jmlr/CarlssonM10}
G.~E. Carlsson and F.~M{\'{e}}moli.
\newblock Characterization, stability and convergence of hierarchical
  clustering methods.
\newblock {\em {JMLR}}, 11, 2010.

\bibitem{DBLP:conf/stoc/CharikarCFM97}
M.~Charikar, C.~Chekuri, T.~Feder, and R.~Motwani.
\newblock Incremental clustering and dynamic information retrieval.
\newblock In {\em {STOC}}, 1997.

\bibitem{DBLP:journals/corr/DeyRS14}
T.~K. Dey, A.~Rossi, and A.~Sidiropoulos.
\newblock Spectral concentration, robust k-center, and simple clustering.
\newblock {\em CoRR}, abs/1404.1008, 2014.
\newblock URL: \url{http://arxiv.org/abs/1404.1008}.

\bibitem{DBLP:conf/esa/DeyRS17}
T.~K. Dey, A.~Rossi, and A.~Sidiropoulos.
\newblock Temporal clustering.
\newblock In {\em {ESA}}, volume~87 of {\em LIPIcs}, 2017.
\newblock URL: \url{https://doi.org/10.4230/LIPIcs.ESA.2017.34}.

\bibitem{DBLP:journals/corr/DeyRS17}
T.~K. Dey, A.~Rossi, and A.~Sidiropoulos.
\newblock Temporal clustering.
\newblock {\em CoRR}, abs/1704.05964, 2017.
\newblock URL: \url{http://arxiv.org/abs/1704.05964}.

\bibitem{DBLP:conf/colt/EldridgeBW15}
J.~Eldridge, M.~Belkin, and Y.~Wang.
\newblock Beyond hartigan consistency: Merge distortion metric for hierarchical
  clustering.
\newblock In {\em {COLT}}, volume~40, 2015.

\bibitem{DBLP:conf/stoc/FarachKW93}
M.~Farach, S.~Kannan, and T.~J. Warnow.
\newblock A robust model for finding optimal evolutionary trees.
\newblock In {\em {STOC}}, 1993.
\newblock URL: \url{http://doi.acm.org/10.1145/167088.167132}.

\bibitem{forgy1965cluster}
E.~W. Forgy.
\newblock Cluster analysis of multivariate data: efficiency versus
  interpretability of classifications.
\newblock {\em Biometrics}, 21, 1965.

\bibitem{DBLP:journals/comgeo/FriedlerM10}
S.~A. Friedler and D.~M. Mount.
\newblock Approximation algorithm for the kinetic robust k-center problem.
\newblock {\em Comput. Geom.}, 43(6-7), 2010.

\bibitem{DBLP:journals/siamcomp/GabowT89}
H.~Gabow and R.~Tarjan.
\newblock Faster scaling algorithms for network problems.
\newblock {\em {SIAM} J. Comput.}, 18(5), 1989.
\newblock URL: \url{https://doi.org/10.1137/0218069}.

\bibitem{DBLP:conf/compgeom/GaoGN04}
J.~Gao, L.~J. Guibas, and A.~Thanh Nguyen.
\newblock Deformable spanners and applications.
\newblock In {\em {SOCG}}, 2004.
\newblock URL: \url{http://doi.acm.org/10.1145/997817.997848}.

\bibitem{DBLP:journals/mor/HochbaumS85}
D.~S. Hochbaum and D.~B. Shmoys.
\newblock A best possible heuristic for the \emph{k}-center problem.
\newblock {\em Math. Oper. Res.}, 10(2), 1985.
\newblock URL: \url{https://doi.org/10.1287/moor.10.2.180}.

\bibitem{DBLP:books/ph/JainD88}
A.~K. Jain and R.~C. Dubes.
\newblock {\em Algorithms for Clustering Data}.
\newblock Prentice-Hall, 1988.

\end{thebibliography}

\end{document}